\documentclass[11pt]{article}

\usepackage[utf8]{inputenc}
\usepackage{graphicx}
\usepackage{xcolor}
\usepackage{xspace}
\usepackage{hyperref}
\usepackage{url}
\usepackage[algo2e,linesnumbered]{algorithm2e}
\usepackage[margin=1in]{geometry}

\usepackage{amsthm}
\usepackage{amsfonts}

\newcommand{\congest}{\ensuremath{\mathsf{CONGEST}}\xspace}
\newcommand{\local}{\ensuremath{\mathsf{LOCAL}}\xspace}
\newcommand{\A}{\mathcal{A}}

\newtheorem{theorem}{Theorem}
\newtheorem{lemma}{Lemma}

\begin{document}

\title{Distributed Reconfiguration of Spanning Trees}

\author{Siddharth Gupta\footnote{University of Warwick, United Kingdom. siddharth.gupta.1@warwick.ac.uk} \and Manish Kumar \footnote{Ben-Gurion University of the Negev, Israel. manishk@post.bgu.ac.il} \and Shreyas Pai\footnote{Aalto University, Finland. shreyas.pai@aalto.fi}}
\date{}
\maketitle

\begin{abstract}
In a reconfiguration problem, given a problem and two feasible solutions of the problem, the task is to find a sequence of  transformations to reach from one solution to the other such that every intermediate state is also a feasible solution to the problem. In this paper, we study the distributed spanning tree reconfiguration problem and we define a new reconfiguration step, called \emph{$k$-simultaneous add and delete}, in which every node is allowed to add at most $k$ edges and delete at most $k$ edges such that multiple nodes do not add or delete the same edge.

We first observe that, if the two input spanning trees are rooted, then we can do the reconfiguration using a single $1$-simultaneous add and delete step in one round in the \congest model. Therefore, we focus our attention towards unrooted spanning trees and show that transforming an unrooted spanning tree into another using a single $1$-simultaneous add and delete step requires $\Omega(n)$ rounds in the \local model. We additionally show that transforming an unrooted spanning tree into another using a single $2$-simultaneous add and delete step can be done in $O(\log n)$ rounds in the \congest model.

\end{abstract}

\newpage

\section{Introduction}
A \emph{reconfiguration problem} asks the following computational question: Given two different configurations of a system, is it possible to transform one to the other in a step-by-step fashion such that the intermediate solutions are also feasible?
Spanning trees are important in classic distributed models such as \local{} and \congest{}\footnote{In the \local\ model \cite{linial92}, a communication network is abstracted as an $n$-node graph. In synchronous \emph{rounds} each node can send an arbitrary size message to each of its neighbors. The \congest\ model \cite{peleg00} is similar to the \local model with the additional constraint that each message has size $O(\log n)$ bits.} as they can be used for efficient routing and aggregation. It is desirable to change the current spanning tree to a better spanning tree depending on the routing demands. Each node can just delete the old incident edges and add the new edges, but this is resource intensive as some nodes may have to simultaneously change a lot of incident edges. Efficiently computing a reconfiguration schedule for spanning trees in a distributed manner allows the system to change from one spanning tree to another in a way that each node is responsible for initiating only a limited amount of changes in one step. And since each intermediate structure is a spanning tree, these intermediate structures can be used to perform the required operations till the next steps are performed.

In the distributed spanning tree reconfiguration problem, we have two spanning trees $T_1, T_2$ of a graph $G$ such that each node $v \in V$ knows its incident edges in $T_1$ and $T_2$. The nodes need to efficiently compute a reconfiguration schedule that converts $T_1$ to $T_2$ using \emph{$k$-simultaneous add and delete steps}, where in each step, each node is allowed to add at most $k$ incident edges to the spanning tree and delete at most $k$ incident edges from the spanning tree. In any given step, multiple nodes cannot add or delete the same edge. A valid reconfiguration schedule is a sequence of steps where we start from $T_1$ and reach $T_2$ such that the intermediate structure obtained after each step is a spanning tree.

If $T_1$ and $T_2$ are rooted spanning trees, where each node knows its parent pointer, then each node $v$ can tell its neighbours that it wants to add its parent in $T_2$ and delete its parent in $T_1$. If $v$ sees that its parent wants to do the opposite operation on the same edge, it does nothing. Hence each edge is added or deleted by at most one node. Therefore, in this setting, the nodes can compute in $1$-round, a reconfiguration schedule using a single $1$-simultaneous add and delete step. But in the case of unrooted trees, this strategy fails as it crucially relies on the parent pointer information to coordinate between the nodes. Therefore, the natural question arises: what can we do in the case of unrooted spanning trees? In this work,\textbf{} we present two results that answer this question:
\begin{enumerate}
    \item A lower bound that shows computing a single step $1$-simultaneous add and delete reconfiguration schedule requires $\Omega(n)$ rounds in the \local model.
    \item An algorithm that computes a single step $2$-simultaneous add and delete reconfiguration schedule in $O(\log n)$ rounds in the \congest model.
\end{enumerate} 

\subsection{Related work}

The problem of spanning tree reconfiguration is very well studied in the centralized setting. A transformation step in the centralized setting is defined as follows: two spanning trees $T$ and $T'$ of a graph $G$ are reachable in one step iff there exists two edges $e \in T$ and $e' \in T'$ such that $T' = (T \setminus e)\cup e'$. In the centralized setting, any spanning tree can be reconfigured into any other spanning tree in polynomial time~\cite{ItoDHPSUU11} and finding a \emph{shortest reconfiguration sequence} between two directed spanning trees is polynomial-time solvable~\cite{ItoIKNOW21}. Therefore, more constrained versions of the problem have been studied. For instance, the reconfiguration problem is PSPACE-complete when each spanning tree in the sequence has \emph{at most (and at least) $k$ leaves} (for $k \geq 3$)~\cite{BousquetI0MOSW20}. On the other hand, reconfiguration is polynomial-time solvable if the intermediate spanning trees are constrained to have large maximum degree and small diameter while it is PSPACE-complete if we have small maximum degree constraints and NP-hard with large diameter constraints~\cite{Bousquet22}. 

The only previous work on distributed spanning tree reconfiguration that we are aware of is~\cite{YamauchiKO21}. In this work, the authors show how to solve reconfiguration of rooted spanning tree in an asynchronous message passing system using local exchange operation between pairs of incident edges, in $O(n)$ rounds and requires $O(\log n)$ bits memory at each process. 
Distributed reconfiguration has been studied for Coloring~\cite{BonamyORSU18,BousquetF21}, Vertex Cover~\cite{HillelMP22}, and MIS~\cite{HillelR20}.

\section{Distributed Spanning Tree Reconfiguration}
In this section, we prove our results for distributed spanning tree reconfiguration. First we show a lower bound of $\Omega(n)$ rounds for computing a single step $1$-simultaneous add and delete reconfiguration schedule in the \local model. And then we present an algorithm that computes a single step $2$-simultaneous add and delete reconfiguration schedule in $O(\log n)$ rounds in the \congest model.

\subsection{$1$-Simultaneous Add and Delete requires $\Omega(n)$ Rounds}

We begin by stating a folk result which will be the basis for the proof of the main lower bound in the subsequent theorem. We state our lower bounds in the \local model, but since any \congest algorithm is also a \local algorithm, the lower bound also holds in the \congest model.
\begin{lemma}\label{lem:rootinglb}
Rooting a tree $T$ at an arbitrary node is a global problem, i.e. it requires $\Omega(n)$ rounds in the \local model.
\end{lemma}
\begin{proof}[sketch]
Let $T$ be a path on $n$ nodes. If each node can output its parent pointer by just looking at the nodes at $o(n)$ radius around it, then we can set the ID's of the nodes in two $o(n)$ length regions of $T$ that are distance $n/10$ apart in such a way that the parents are pointing away from each other in the path connecting the two neighbourhoods. This forces either more than one root or one node to have more than one parent pointers, a contradiction.
\end{proof}

\begin{figure}[tb]
\centering 
\includegraphics[page=1,width=0.9\textwidth]{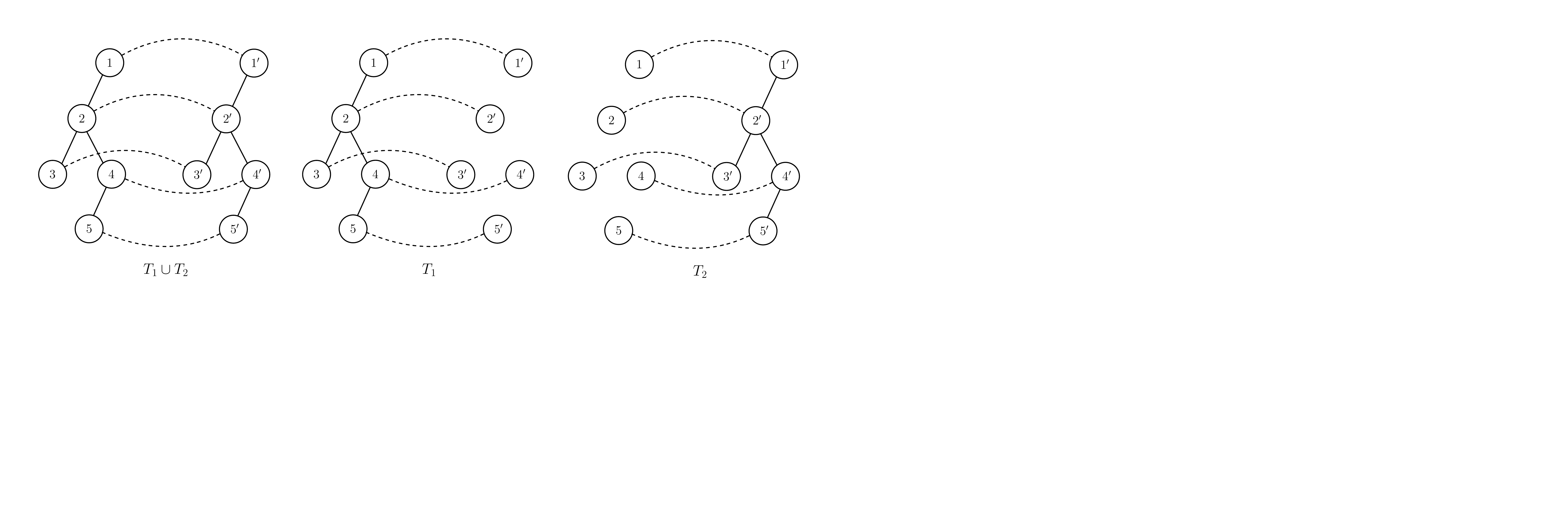}
\caption{This figure illustrates the reduction from rooting to spanning tree reconfiguration. The input $T$ to the rooting problem is the tree spanned by nodes $1$-$5$ in the left graph $T_1 \cup T_2$, which is the input graph we create for the reconfiguration problem. The source spanning tree $T_1$ is in the middle and the destination spanning tree $T_2$ is on the right.}
\label{f:Figure1}
\end{figure}

\begin{theorem}
Solving the distributed spanning tree reconfiguration problem in one step of $1$-simultaneous add and delete requires $\Omega(n)$ rounds in the \local model.
\end{theorem}
\begin{proof}
For sake of contradiction let $A$ be a \local algorithm that computes a one step reconfiguration schedule in $o(n)$ rounds. We will show that $A$ can be used to root an unrooted tree in $o(n)$ rounds in the \local model. Let $T = (V, E)$ be the tree that we wish to root. For each node $v \in V$ create a copy $v'$ which will be simulated by $v$, and add edges $\{v, v'\}$ as well as $\{u', v'\}$ for all neighbours $u$ of $v$. Let $V'$ be the set of all the nodes $v'$, $E'$ be the set of edges of the form $\{u', v'\}$, and $M = \{\{v, v'\} \mid v \in V\}$. Now we want to run algorithm $\A$ where the source spanning tree is $T_1 = (V \cup V', E \cup M)$, the destination spanning tree is $T_2 = (V \cup V', E' \cup M)$, and the communication network is $T_1 \cup T_2$. An example of this reduction is shown in Figure~\ref{f:Figure1}.

Any $R$-round \local algorithm on $T_1 \cup T_2$ can be simulated on network $T$ in $R$ rounds by having $v$ simulate the behaviour of $v'$. Since $\A$ produces a reconfiguration schedule that uses one step of $1$-simultaneous add and delete, each node will delete at most one edge and add at most one edge in order to go from $T_1$ to $T_2$. Node $v$ will output as its parent the edge in $E$ that is to be deleted by $v$, if such an edge exists.

These parent pointers correspond to a valid rooting because nodes in $V$ must delete $n-1$ edges of $T_1$ in one step for the reconfiguration schedule of $\A$ to be correct. This is only possible if $n-1$ nodes of $T$  delete exactly one incident edge of $T$ and the remaining node $r$ does not delete any incident edge. The neighbours of $r$ in $T$ must delete the incident edge that is pointing to $r$ as nobody else can delete this edge. Then we can repeat this argument inductively on all nodes that are $i$-hops away from $r$ and we show that the parent pointers form a valid rooting of $T$ with root node $r$.

Thus, a rooting of $T$ was output in the \local model in $o(n)$ rounds, which is impossible by Lemma~\ref{lem:rootinglb}. Thus $\A$ cannot exist, which proves the theorem.
\end{proof}

\subsection{$2$-Simultaneous Add and Delete in $O(\log n)$ Rounds}
We first describe an edge orientation process that is essentially the rake and compress algorithm\footnote{See \url{https://discrete-notes.github.io/rake-and-compress} for an nice description of the rake and compress algorithm} 
of~\cite{DBLP:journals/acr/MillerR89}. The output of Algorithm~\ref{alg:peeling} has the property that each node has at most two outgoing edges.
It is well known that this orientation can be computed very efficiently as opposed to a rooting of the tree.

\RestyleAlgo{boxruled}
\begin{algorithm2e}\caption{\textsc{Orient}$(T)$\label{alg:peeling}}
$T' \leftarrow $ empty graph \\
\While{$T \neq \emptyset$} {
    $H \leftarrow $ nodes in $T$ with degree at most $2$\\
    Add to $T'$ the nodes of $H$ with their incident edges in $T$ oriented outward, breaking ties arbitrarily \\
    Remove $H$ from $T$
}
\Return\ oriented tree $T'$
\end{algorithm2e}

\begin{lemma}
The while loop of Algorithm~\ref{alg:peeling} runs for $O(\log n)$ iterations. Moreover each iteration can be implemented in $O(1)$ rounds in the \congest model.
\end{lemma}
\begin{proof}
If $\alpha$ fraction of the nodes have degree at most $2$ then we can write average degree as at least $3(1-\alpha)$ because $(1-\alpha)$ fraction of the nodes must have degree at least $3$. Average degree of a tree (or a forest) is $2$, which implies $2 \ge 3(1-\alpha)$. So an $\alpha \ge 1/3$ fraction of the nodes are removed in each iteration. Therefore the number of iterations is at most $O(\log n)$.

To execute an iteration in the \congest model, each node just needs to know its degree in the current tree $T$. So in the each iteration, nodes in $H$ that remove themselves can send a message to their neighbours to decrease their degree.
\end{proof}
\begin{figure}[tb]
\centering 
\includegraphics[page=2,width=0.9\textwidth]{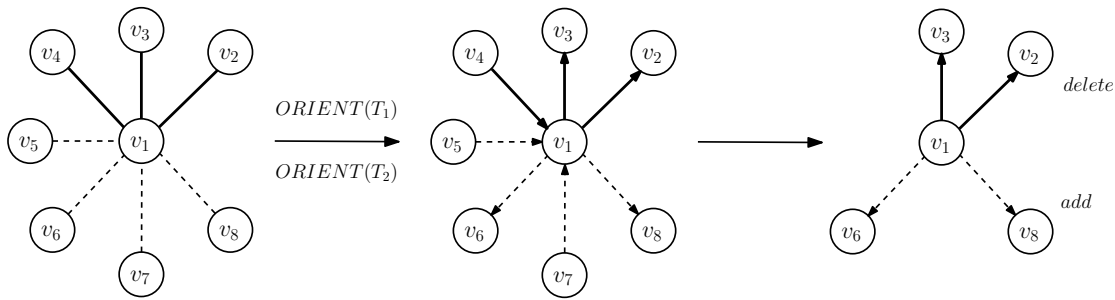}
\caption{This figure shows how node $v_1$ computes reconfiguration schedule using the orientation procedure. The bold edges belong to $T_1$ and the dashed edges belong to $T_2$}
\label{f:Figure2}
\end{figure}
Now we show how this orientation can be used to compute a reconfiguration schedule. We run Algorithm~\ref{alg:peeling} on the source spanning tree $T_1$ and the target spanning tree $T_2$ separately and obtain two orientations such that each node in the graph has at most $2$ outgoing edges in $T_1$ and at most $2$ outgoing edges in $T_2$ (see  Figure~\ref{f:Figure2}). Now, each node $v$ decides it will add its outgoing edges of $T_2$ and it will delete its outgoing edges of $T_1$. If $v$ decides to add and delete the same incident edge $e$, it updates its decision to not change $e$. Then $v$ sends the decisions along the outgoing edges. If for a single edge, one end point has decided to add and the other has decided to delete, then both nodes update their decision on this edge to do nothing. It is easy to see that these decisions form one step of $2$-simultaneous add and delete where all edges of $T_1 \setminus T_2$ are deleted and all edges of $T_2 \setminus T_1$ are added. Therefore, we have computed a one step reconfiguration schedule in $O(\log n)$ rounds of the \congest model.
This proves the following theorem.

\begin{theorem}
The distributed spanning tree reconfiguration problem can be solved in one step of $2$-simultaneous add and delete. Computing this reconfiguration schedule takes $O(\log n)$ rounds in the \congest model.
\end{theorem}

\section{Future Work}
Our results only pertain to computing \emph{single step} schedules for the distributed spanning tree reconfiguration problem. A natural next question is what happens if we want to compute reconfiguration schedules with multiple steps? In particular, how many rounds do we need to compute a multi-step $1$-simultaneous add and delete reconfiguration schedule?

While our lower bound of $\Omega(n)$ rounds for computing a single step $1$-simultaneous add and delete reconfiguration schedule is tight, we don't know if our $O(\log n)$ round algorithm for computing a single step $2$-simultaneous add and delete reconfiguration schedule is optimal. Can we show a matching lower bound or design a faster algorithm in this case?

\bibliographystyle{plain}
\bibliography{mybibliography}

\begin{thebibliography}{10}

\bibitem{BonamyORSU18}
Marthe Bonamy, Paul Ouvrard, Mika{\"{e}}l Rabie, Jukka Suomela, and Jara Uitto.
\newblock Distributed recoloring.
\newblock In {\em DISC 2018}, 2018.

\bibitem{BousquetF21}
Nicolas Bousquet, Laurent Feuilloley, Marc Heinrich, and Mika\"{e}l Rabie.
\newblock {Distributed Recoloring of Interval and Chordal Graphs}.
\newblock In {\em OPODIS 2021}, 2021.

\bibitem{BousquetI0MOSW20}
Nicolas Bousquet, Takehiro Ito, Yusuke Kobayashi, Haruka Mizuta, Paul Ouvrard,
  Akira Suzuki, and Kunihiro Wasa.
\newblock Reconfiguration of spanning trees with many or few leaves.
\newblock In {\em ESA 2020}, 2020.

\bibitem{Bousquet22}
Nicolas Bousquet, Takehiro Ito, Yusuke Kobayashi, Haruka Mizuta, Paul Ouvrard,
  Akira Suzuki, and Kunihiro Wasa.
\newblock Reconfiguration of spanning trees with degree constraint or diameter
  constraint.
\newblock In {\em STACS 2022}, 2022.

\bibitem{HillelMP22}
Keren Censor{-}Hillel, Yannic Maus, Shahar~Romem Peled, and Tigran Tonoyan.
\newblock Distributed vertex cover reconfiguration.
\newblock In {\em ITCS 2022}, 2022.

\bibitem{HillelR20}
Keren Censor{-}Hillel and Mika{\"{e}}l Rabie.
\newblock Distributed reconfiguration of maximal independent sets.
\newblock {\em J. Comput. Syst. Sci.}, 112:85--96, 2020.

\bibitem{ItoDHPSUU11}
Takehiro Ito, Erik~D. Demaine, Nicholas J.~A. Harvey, Christos~H.
  Papadimitriou, Martha Sideri, Ryuhei Uehara, and Yushi Uno.
\newblock On the complexity of reconfiguration problems.
\newblock {\em TCS}, 2011.

\bibitem{ItoIKNOW21}
Takehiro Ito, Yuni Iwamasa, Yasuaki Kobayashi, Yu~Nakahata, Yota Otachi, and
  Kunihiro Wasa.
\newblock Reconfiguring directed trees in a digraph.
\newblock In {\em COCOON 2021}, 2021.

\bibitem{linial92}
Nati Linial.
\newblock Locality in distributed graph algorithms.
\newblock {\em SIAM J. on Computing}, 1992.

\bibitem{DBLP:journals/acr/MillerR89}
Gary~L. Miller and John~H. Reif.
\newblock Parallel tree contraction part 1: Fundamentals.
\newblock {\em Adv. Comput. Res.}, 5:47--72, 1989.

\bibitem{peleg00}
D.~Peleg.
\newblock {\em Distributed Computing: A Locality-Sensitive Approach}.
\newblock SIAM, 2000.

\bibitem{YamauchiKO21}
Yukiko Yamauchi, Naoyuki Kamiyama, and Yota Otachi.
\newblock Distributed reconfiguration of spanning trees.
\newblock In {\em SSS 2021}, 2021.

\end{thebibliography}

\end{document}